\documentclass[a4paper, 12pt]{article}

\usepackage{inputenc}
\usepackage{amsfonts}
\usepackage{amsmath}
\usepackage{amssymb}
\usepackage{amsthm}
\usepackage{booktabs}
\usepackage{bbm}
\usepackage{float}
\usepackage{multirow}
\usepackage{graphicx} 
\usepackage{graphics}
\usepackage[semicolon,authoryear]{natbib}
\usepackage{longtable}
\usepackage[]{threeparttable}
\usepackage{rotating}
\usepackage[dvipsnames]{xcolor}
\usepackage{comment}
\usepackage{tikz}
\usetikzlibrary{positioning}
\usepackage{textcomp}
\usepackage{makecell}
\usepackage{xurl}
\usepackage{float}
\usepackage{soul}
\usepackage[ruled,vlined,linesnumbered]{algorithm2e}

\newtheorem{theorem}{Theorem}
\newtheorem{corollary}{Corollary}

\usepackage{hyperref}
\usepackage[noabbrev,capitalize]{cleveref}
\usepackage{marvosym}

\crefname{algocf}{algorithm}{Algorithms}

\title{Robust confidence intervals for generalized linear models}

\author{Andrea Panarotto$^{1,*}$,
Riccardo De Santis$^{1}$, and Livio Finos$^{1,2}$\\
\small
$^{1}$Department of Statistical Sciences, University of Padova, Padova, Italy\\
\small
$^{2}$Padova Neuroscience Center, University of Padova, Padova, Italy\\
\small
\Letter andrea.panarotto@unipd.it
	   }
\date{}
\begin{document}

\maketitle

\begin{abstract}
Reliable uncertainty quantification is a central challenge in the analysis of modern biomedical data, where complex sources of variability often violate standard modeling assumptions. In generalized linear models (GLMs), confidence intervals for regression parameters provide such information, but they typically rely on correct specification of the mean–variance relationship. However, overdispersion, heteroskedasticity, and unobserved biological variability can lead to substantial undercoverage in practice. We propose a method for constructing confidence intervals that remains valid under variance misspecification. The approach is based on the inversion of hypothesis tests obtained by sign-flipping individual score contributions and uses a bisection algorithm to determine the interval bounds. The resulting intervals inherit robustness properties from the underlying tests, and we establish their asymptotic validity under general variance misspecification. Through simulation studies, we show that the proposed method achieves reliable coverage and outperforms standard Wald-type intervals when model assumptions are violated. We illustrate the approach in a differential expression analysis of RNA-sequencing data from a cancer study, where heterogeneous variability is pervasive and parametric methods can yield inconsistent inference. The proposed framework provides a practical and robust alternative to conventional quasi-likelihood or Wald-based methods for interval estimation in GLMs, particularly suited to high-throughput biomedical applications.

\end{abstract}

\section{Introduction}
\label{sec:intro}

Confidence intervals play a central role in statistical inference, providing uncertainty quantification that complements point estimates and supports principled scientific conclusions. In generalized linear models (GLMs), confidence intervals are routinely used to assess the magnitude and relevance of covariate effects across a wide range of applications, from biomedical studies to economics and the social sciences \citep[e.g.][]{mcgrath_green_2018,rathnayake_bootstrapping_2023,xia_debiased_2023,kang_accurate_2023,donayre_likelihood-ratio-based_2025,chowdhury_improving_2026}. 

Standard approaches to interval construction in GLMs are inherently parametric and rely on the correct specification of the underlying model, since they are based on a full-likelihood approach \citep[Section 4.3]{agresti2015foundations}. When these assumptions are violated (as is common in practice due to heteroskedasticity, overdispersion, or unobserved sources of variability) classical confidence intervals may suffer from severe undercoverage, leading to inaccurate inference. Although robust covariance estimators, such as sandwich-type estimators \citep{cox1961tests, cox1962further, Huber1967, White1982}, can partially alleviate these issues in large samples, their performances show often a problematic slow convergence to the nominal level \citep[see for instance the simulation study in][which focuses on hypothesis testing, and references therein]{de2025inference}.

Despite the prevalence of variance misspecification in applied settings, relatively few methods are available to construct confidence intervals in GLMs that are robust to such violations. In simple settings, nonparametric testing procedures, as permutation and randomization tests, are well known for their robustness properties under weak distributional assumptions. Since confidence intervals can be obtained by inversion of hypothesis tests, this naturally suggests the development of interval estimators that inherit the robustness of the underlying tests \citep{casella1984confidence,pesarin_multivariate_2001}. However, they are not easy to implement when nuisance confounders are present, as in regression-type models, since the sample units are not more exchangeable. 
Two fully nonparametric approaches, based on sample-splitting, have been proposed, respectively, by \cite{wasserman2020universal} and \cite{kuchibhotla2024hulc}. However, these methods can easily lead to overly large confidence intervals and to conservative inference \citep{tse2022note, de2025inference}.

In this work, we propose a method for constructing confidence intervals in GLMs based on the inversion of a resampling-based test, that is, the sign-flip score test \citep{hemerik2020robust,de2025inference}. This test is based on a semi-parametric approach, and have recently attracted attention due to its finite-sample behavior and robustness to general variance misspecification under mild assumptions. Exploiting these properties, the proposed confidence intervals remain reliable even when the parametric variance structure is incorrectly specified.

However, constructing confidence intervals by test inversion presents additional challenges that are often overlooked in the literature. In particular, p-value functions based on permutation and randomization tests are inherently discrete and typically stepwise, as noted by \cite{glazer_fast_2025}. Moreover, such p-value functions might not be monotone in the parameter of interest. Although this might not be a problem from the formal point of view, it is of practical importance, complicating the inversion procedure and potentially compromising the intended nominal coverage. Many existing approaches to confidence intervals derived from permutation tests do not explicitly address these issues \citep[e.g.,][]{garthwaite_confidence_1996,pagano_obtaining_1983} or need further assumptions to ensure such monotonicity: \cite{pesarin_multivariate_2001}, for example, assumes the same variance for all hypotheses tested.

We propose a principled and practically implementable framework for robust confidence interval estimation in GLMs. We provide finite-sample and asymptotic monotonicity results for p-value functions arising from sign-flip score tests under suitable conditions. Building on these theoretical guarantees, we introduce a bisection algorithm for interval construction that explicitly accounts for discreteness and ensures asymptotically correct nominal coverage, overcoming the difficulties mentioned above.

The remainder of the paper is organized as follows. Section 2 reviews the sign-flip score tests and their main properties. Section 3 presents the proposed confidence interval construction, including the monotonicity results and the bisection algorithm. Sections 4 and 5 compare the proposed method with classical parametric approaches, first in simulation studies under variance misspecification and then in an application to RNA-sequencing data from a cancer genomics study. Concluding remarks are given in Section 6.

\section{Sign-flip score tests}
\label{sec:sign_flip}

In this section we recap the main ideas behind the flipscores approach as detailed in \cite{de2025inference}, in the ambit of Generalized Linear Models (GLMs).

Let us consider a problem where we observe $n$ independent observations $y=\{y_1,\dots,y_n\}^T$ realization of the random variable $Y$. Note that the dependence on $n$ of the vectors and matrices will be generally suppressed to help the readability. Each $i$-th observation is assumed to have a density belonging to the exponential dispersion family, \textit{i.e.} a density of the form \citep{agresti2015foundations}
\begin{equation} \label{eq:model}
    f(y_i; \theta_i, \phi_i)=
    \exp\left\{\frac{y_i\theta_i-b(\theta_i)}{\phi_i}+c(y_i,\phi_i)\right\}.
\end{equation}
Here $\theta_i$ and $\phi_i$ are respectively the canonical and the dispersion parameter, while $b(\cdot)$ is a known function. In this setting, we derive for the random variable $Y_i$, related to the observed value $y_i$,
\begin{equation}
    \mu_i=\mathbb{E}(Y_i)=b'(\theta_i)\,, \qquad
\mathbb{V}(Y_i)=b''(\theta_i)\phi_i\,.
\label{eq:glm_mv}
\end{equation}
The $\mu_i,\ i=1,\dots,n$ are assumed to depend on some observed covariates $(x_i,\mathbf{z}_i)$ through the equation
\begin{equation}
    g(\mu_i) = \eta_i = x_i\beta+\mathbf{z}_i^T\boldsymbol{\gamma}\,.
    \label{eq:glm_link}
\end{equation}
We define $X=\{x_i\}_{i=1}^n$ as the $n$-vector of the target covariate and $Z=\{\mathbf{z}_i\}_{i=1}^n$ as the $n \times p$-matrix of nuisance covariates. The regression parameters $(\beta,\boldsymbol{\gamma})$ are, respectively, a scalar and a $p$-dimensional parameters. Remarkably, no restrictions are posed on the dispersion parameters. We will assume through the paper the conditioning on the observed covariates, as is usual in regression-based analyzes. Moreover, we will assume the same technical assumptions as in \cite{de2025inference}, without further details. These assumptions should not be considered as limiting the proposal of this manuscript, but are made to avoid tricky and counter-intuitive situations.

\cite{de2025inference} propose a test for 
\[
    H_0:\beta=\beta_0
\]
against any pre-determined general alternative. The test is based on the idea of sign-flips, in the following way.

Let $S$ be the first derivative of the likelihood function (known as the score vector), and $\mathcal{I}$ be the second derivative (known as Fisher information matrix). They can be partitioned as 
\[
    S=\begin{pmatrix}
        s_\beta & s_\gamma
    \end{pmatrix}^T
\]
and
\[
    \mathcal{I}=
    \begin{pmatrix}
    \mathcal{I}_{\beta,\beta} & \mathcal{I}_{\beta,\boldsymbol{\gamma}} \\
    \mathcal{I}_{\boldsymbol{\gamma},\beta} &
    \mathcal{I}_{\boldsymbol{\gamma},\boldsymbol{\gamma}}
    \end{pmatrix}
    =
    \begin{pmatrix}
    X^T W X & 
    X^T W Z \\
    Z^T W X &
    Z^T W Z
    \end{pmatrix}
\]
where, given $\ell(\cdot)$ as the likelihood function,
\[
    \begin{array}{c}
    \dfrac{\partial \ell(\beta,\boldsymbol{\gamma})}{\partial \beta}=s_\beta=X^T D V^{-1}(Y-\boldsymbol{\mu})\,, \\
    \dfrac{\partial \ell(\beta,\boldsymbol{\gamma})}{\partial \boldsymbol{\gamma}}=s_{\boldsymbol{\gamma}}=Z^T D V^{-1}(Y-\boldsymbol{\mu})\,, \\
    D=\mathrm{diag}\left\{\dfrac{\partial \mu_i}{\partial \eta_i}\right\}\,, \quad V=\mathrm{diag}\{\mathbb{V}(y_i)\}\,,
    \end{array}
\]
and $W=D V^{-1} D$.

\cite{hemerik2020robust} define the effective score as
\[
S(\beta,\boldsymbol{\gamma}) = s_\beta - \mathcal{I}_{\beta, \boldsymbol{\gamma}} \mathcal{I}_{\boldsymbol{\gamma}, \boldsymbol{\gamma}}^{-1} s_{\boldsymbol{\gamma}}\,.
\] 
In the context of generalized linear models, the statistic can be written as
\[
    S(\beta,\boldsymbol{\gamma})=n^{-1/2}X^TW^{1/2}(I-H) V^{-1/2}(Y-\boldsymbol{\mu})
\]
where
\[
\label{eq:hproj}
    H=W^{1/2}Z(Z^TW Z)^{-1}Z^TW^{1/2}
\]
is the hat matrix. Note that $S$ depends on the unknown regression parameters.

As usual for score-type tests, the test of \cite{de2025inference} is obtained computing the test statistic under the model implied by the null hypothesis, which we will denote as null model. In such case, $\beta_0$ enter as an offset term in the estimation procedure. We will denote $S=S(\beta_0,\hat{\boldsymbol{\gamma}})$, and $\hat{\mu}$ the estimate of $\mu$ based on the null model; note that it is defined as a sum of $n$ elements. The flipscores test consists in randomly changing the sign of such elements many different times. If we define, for each flip, a diagonal matrix $F_g$ of $-1$ and $1$, we get a vector of test statistics $S(I),S(F_2),\dots,S(F_G)$ where 
\begin{equation}
    S(F_g)=n^{-1/2}X^T{W}^{1/2}(I-{H})F_g{V}^{-1/2}(Y-\hat{\boldsymbol\mu}).
    \label{eq:eff_score_stat}
\end{equation}
Note that we will denote with $I$ the $n$-dimensional identity matrix. Practically, given that the total of possible flips is $G=2^n$, a subset of such flips is always used -- whose amount will be denoted with $w$. This is theoretically justified in \cite{hemerik2018exact}.

While being already usable, an improvement (in terms of convergence to the nominal level of the test) from the cited effective score test (proposed by \cite{hemerik2020robust}) is obtained in \cite{de2025inference}, which derive the standardized score statistic as
\begin{equation}
    S^*(F_g) = S(F_g)/\mathbb{V}(S(F_g))^{1/2},
    \label{eq:def_standardization}
\end{equation}
where 
\[ 
    \mathbb{V}(S(F_g)) = n^{-1}X^T W^{1/2} (I-H)F_g(I-H)F_g (I-H)W^{1/2}X.
\]
A test using the standardized or the effective score statistics can be built in the following theorem. This coincides with Theorem 2 of \cite{de2025inference} and, remarkably, does not require the correct specification of the variance; that is, we can do reliable inference on the regression parameters regardless of the misspecification of the dispersion parameters. 

\begin{theorem}
For every $1\leq g \leq w$, consider the statistic $T_g^n=S^*(F_g) $ (or $S(F_g)$) and let $T^n_{(1)}\leq ...\leq T^n_{(w)}$ be the sorted test-statistics. Consider the test that rejects if $T_1^n > T^n_{\lceil(1-\alpha)w\rceil}$. 
As $n\rightarrow \infty$, under $H_0$ the rejection probability converges to 
$\lfloor \alpha w \rfloor /w \leq \alpha$. 
\end{theorem}
The p-value associated to the test is $(r + 1) / (w+1)$, where $r=\sum_{g=1}^w 1\left\{T_1^{n} > T_g^{n}\right\}$.

\section{Confidence intervals in flipscores}
Confidence sets are obtained by inverting hypothesis tests \citep{casella1984confidence,pesarin_multivariate_2001}. A $(1-\alpha)$-confidence interval for a parameter $\beta$ is defined as the set of all parameter values $\beta_0$ that would not be rejected by a level-$\alpha$ test of the null hypothesis 
\[
    H_0: \beta = \beta_0
\]
Formally:
    \begin{equation}
        C_{1-\alpha} (X) = \left\{ \beta_0\mid \text{p-value}(X;H_0: \beta= \beta_0)  \geq \alpha\right\}
        \label{eq:test_inversion}\,.
    \end{equation}
This construction guarantees the nominal coverage property of confidence intervals.

Therefore, we fix a confidence level $1 - \alpha$ and, once we obtain the estimate $\hat{\beta}_{\text{obs}}$ of $\beta$, we consider and invert the one-sided tests
\begin{equation}    
    H_0:\ \beta = \beta_0 \quad \text{vs} \quad H_1:\ \beta > \beta_0\,
    \label{eq:one-sided-lower}
\end{equation}
for $\beta_0 < \hat{\beta}_{\text{obs}}$, and
\begin{equation}    
    H_0:\ \beta = \beta_0 \quad \text{vs} \quad H_1:\ \beta < \beta_0\,
    \label{eq:one-sided-upper}
\end{equation}
for $\beta_0 > \hat{\beta}_{\text{obs}}$. We distinguish these two cases to simplify the search for the confidence bounds; in this way, the p-value relative to test (\ref{eq:one-sided-lower}) (resp. (\ref{eq:one-sided-upper})) is maximum in $\beta_0 = \hat{\beta}_{\text{obs}}$ and tends to decrease as $\beta_0$ becomes smaller (resp. larger), that is, more distant from the estimate. So, we invert test (\ref{eq:one-sided-lower}) at significance level $\alpha/2$ to find the lower bound of the confidence interval, and test (\ref{eq:one-sided-upper}) at the same level to find the upper bound. {Moreover, the p-values relative to the test (\ref{eq:one-sided-lower}) with $\beta_0>\hat{\beta}_{\text{obs}}$ (or (\ref{eq:one-sided-upper}) with $\beta_0<\hat{\beta}_{\text{obs}}$) are generally higher than the one where $\beta_0 = \hat{\beta}$, which is around 0.5 and will never be lower than any sensible choice for $\alpha/2$}. From now on, let us restrict ourselves to the search for the lower limit, as everything is analogous for the upper one.

We call $f_p\left(\beta_0\right)$ the function that associates to each $\beta_0 \leq \hat{\beta}_{\text{obs}}$ the p-value from the corresponding corresponding test (\ref{eq:one-sided-lower}). We wish for a monotonicity property for $f_p$. In fact, if $f_p$ were non-decreasing, we would be able to find a unique $\beta_L$ such that $\beta_0 \notin CI$ for $\beta_0 < \beta_L$ and $\beta_0 \in CI$ for $\beta_L \leq \beta_0 \leq \hat{\beta}_{\text{obs}}$. A simple search algorithm, such as a bisection algorithm, would guarantee that we get $\beta_L$ (up to a certain tolerance) and that the corresponding confidence set would be a connected interval. 

This monotonicity property holds for parametric tests and, as proved by \cite{pesarin_multivariate_2001}, for permutation tests when the observations are fully exchangeable. In our framework, this means that the only nuisance parameter allowed would be the intercept. We prove that monotonicity holds also when employing the effective flipscores test on linear models, in presence of any possible set of confounders.

\begin{theorem}
    \label{thm:monot_lm_eff}
    Assume that the observations are independent and generated by the linear model 
    \[
        Y_i = x_i\beta + \mathbf{z}_i^T \boldsymbol{\gamma} + \varepsilon_i\,,\quad \varepsilon_i \sim \mathcal{N}(0, \sigma^2)\, ,
    \]
    for $i = 1,\ldots,n$.
    Let $\hat{\beta}_{\text{obs}}$ be the maximum likelihood estimate of $\beta$ and let $\beta_A < \beta_B < \hat{\beta}_{\text{obs}}$. Let $p_k$, $k=A,B$, be the p-values associated with the effective flipscores tests $H_0^{(k)}:\beta=\beta_k$ versus $H_1^{(k)}: \beta > \beta_k$, for $k = A,B$, respectively, assuming that the tests share the same set of flip matrices. Then, $p_A\leq p_B$.
\end{theorem}

When employing the standardized flipscores test, monotonicity is not given in the finite case. However, we can prove asymptotic monotonicity for linear models. From \cite{de2025inference} (in particular, the proof of Theorem 2, in their Appendix), we have, under the null hypothesis, the following property for the variance of the flipped score statistics:
\begin{equation}
    \lim_{n\rightarrow\infty} \mathbb{V}\left(S(I)\right) - \mathbb{V}(S({F_g})) = 0\,,
    \label{eq:asy_equi_var_flipscores}
\end{equation}
for any $F_g$ and $\beta$. This leads to the following corollary.
\begin{corollary}
    Assume that the observations are independent and generated by the linear model 
    \[
        y_i = x_i \beta + \mathbf{z}_i^T \boldsymbol{\gamma} + \varepsilon_i\,,\quad \varepsilon_i \sim \mathcal{N}(0, \sigma^2)\, ,
    \]
    for $i=1,\ldots,n$. 
    Let $\hat{\beta}_{\text{obs}}$ be the estimate of $\beta$ and let $\beta_A < \beta_B < \hat{\beta}_{\text{obs}}$. Let $p_k$, $k=A,B$, be the p-values associated with the standardized flipscore tests $H_0^{(k)}:\beta=\beta_k$ versus $H_1^{(k)}: \beta > \beta_k$, for $k = A,B$, respectively. Then, as $n\rightarrow\infty$, it holds that $p_A\leq p_B$.
\end{corollary}

The extension to generalized linear models requires further steps. In fact, when the null $\beta_0$ changes, the offset is changed, then all $\eta_i$s are affected, and, finally, the means $\mu_i$s change with a non-linear dependence. Outside of the linear model and other cases where the variance does not depend on the mean, this means that the offset will modify the estimate of the weights $W$ in each score. However, we still recover some properties in the asymptotic framework.

\begin{theorem}
\label{thm:monot_glm_asy}
Assume that the observations are independent and generated from a generalized linear model defined by Equations (\ref{eq:model})--(\ref{eq:glm_link}). Let $\hat{\beta}_{obs}$ denote the estimate of $\beta$ and let $\beta_A < \beta_B < \hat{\beta}_{obs}$. Let $S^{(k)}(F_g)$ denote the effective or standardized sign-flip score statistic computed under $H_0^{(k)}:\beta=\beta_k$ with flip matrix $F_g$, and let $S^{(k)}(I)$ denote the observed statistic (identity flip). 
Then, for any fixed flip matrix $F_g$,
\begin{equation*}
        \Pr\left(S^{(A)}(F_g) < S^{(A)}(I) \mid S^{(B)}(F_g) < S^{(B)}(I) \right) \;\longrightarrow\;1
\end{equation*}
as $n\to\infty$. Consequently, the sign-flip p-value for testing $H_0^{(A)}$ is asymptotically no larger than that for testing $H_0^{(B)}$.
\end{theorem}

The proofs to the theorems in this section can be found in Appendix A.
    
\subsection{Monotonicity example and counter-example}
\cref{fig:monotonicity} provides examples of a non-decreasing p-value function and of a non-monotonic p-value function, on simulated observations. The datasets are built by simulating $n$ observations from the logistic model
\begin{equation}
    Y_i\sim Bernoulli(\pi_i)\,,\quad \pi_i = \frac{\exp(\eta_i)}{1 + \exp(\eta_i)}\,, \quad \eta_i = \beta x_i + \gamma z_i\,,
    \label{eq:monot_ex_model}
\end{equation}
where $(\beta,\gamma) = (0,-0.5)$ and
\begin{equation}
    (X_i,Z_i) \sim N_2(\mathbf{0}, \Sigma)\,, \qquad \Sigma = \begin{bmatrix}
        1 & 0.2 \\
        0.2 & 1
    \end{bmatrix}.
    \label{eq:monot_ex_covs}
\end{equation}
The number of observations is $n=50$ for the monotonic example, and $n=20$ for the non-monotonic. Test (\ref{eq:one-sided-lower}) is performed as $\beta_0$ varies on an uniform grid to provide the p-values.

In the monotonic case, that the lower bound $\beta_L$ of the confidence interval is obtained when $f_p$ reaches the value $\alpha/2$, ensuring the efficacy of a bisection algorithm. In the non-monotonic case, the bisection algorithm does not ensure the retrieval of the correct lower bound $\beta_L$, but might converge to a bound $\tilde{\beta}_L$ such that not all points $\beta_1 < \beta_L$ imply $f_p(\beta_1) < \alpha/2$. The black points between $\beta_L$ and $\tilde{\beta}_L$ are wrongly excluded from the confidence interval. This justifies the need for the theoretical results in \cref{thm:monot_lm_eff,thm:monot_glm_asy}.

    
\begin{figure}
    \centering
    \includegraphics[width=\linewidth]{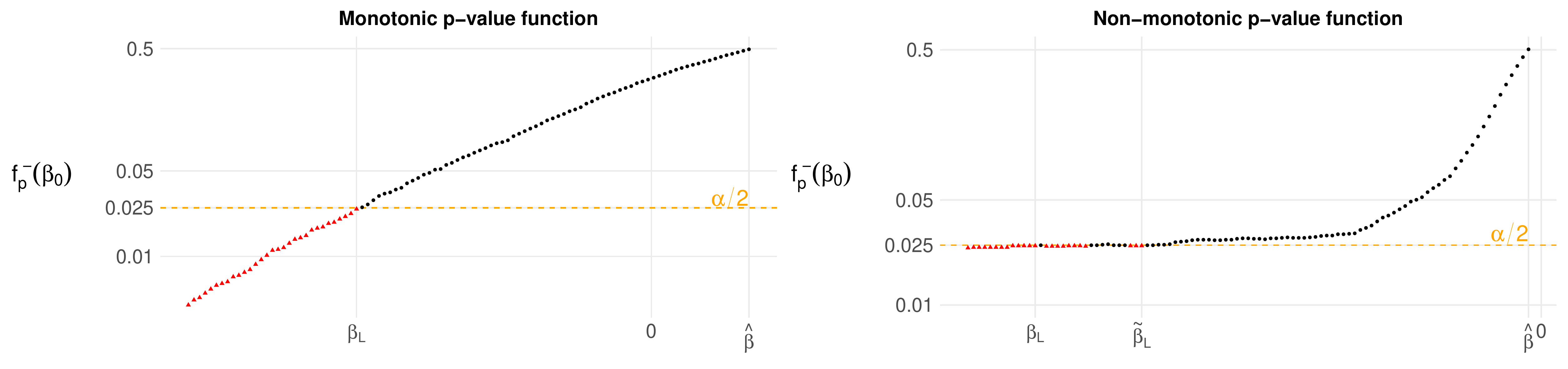}
    \caption{Values of $f_p(\beta_0)$ on an equispaced grid of values between $\hat{\beta}_{\text{obs}} - 1$ and $\hat{\beta}_{\text{obs}}$. The $n$ observations have been generated from a logistic model. The figure shows in red the points with an associated p-value lower than $\alpha/2$, and in black the points that should belong to the confidence set. With $n = 50$ (left) the p-value function is monotonic. With $n = 20$ (right), the function is non-monotonic and the bisection procedure might lead to inaccurate inference.}
    \label{fig:monotonicity}
\end{figure}

However, we should mention that, in practice, these problems are not so relevant: the monotonicity property we seek is obtained for relatively small $n$ (we already struggle to find counter-examples for $n = 25$), so nominal coverage is easily ensured, as shown in the simulations in \cref{sec:simul}. Moreover, since a limited number of random flip matrices is adopted, the use of the same flips across the iterations of the bisection algorithm reduces further the probability of stumbling into non-monotonicity situations.

\subsection{Bisection search algorithm}

In this section, we propose the algorithm for the search of the confidence interval bounds. Let $f_p^-(\beta_0)$ be the function that returns the p-value for the right-sided test, with $\beta_0 < \hat{\beta}$, and let $f_p^+(\beta_0)$ be the analogous function that returns the p-value for the left-sided test with $\beta_0 > \hat{\beta}$. We are assuming $f_p^-$ to be non-decreasing and $f_p^+$ to be non-increasing in the respective domains. We propose two options, producing respectively what we call an ``equitailed" confidence interval and a symmetric confidence interval. Both algorithms are based on a bisection-like search.

For the \textit{equitailed confidence interval}, we look for $\beta_L,\ \beta_U$ such that 
\begin{equation}
    \begin{aligned}   
    \beta_L = \inf \left\{\beta_0 < \hat{\beta}:f_p^-(\beta_0)\geq\alpha/2\right\},\\
    \beta_U = \sup \left\{\beta_0>\hat{\beta}:f_p^+(\beta_0)\geq\alpha/2\right\}.
    \end{aligned}
    \label{eq:bound_definition}
\end{equation}

Again, we will focus on the left side as the right part is analogous. To start our bisection procedure, we need to restrict our search to a starting interval. The upper bound of the interval could be $\hat{\beta}$, but the lower bound is not so easy to select, as we need to find a value $\tilde{\beta}$ such that $f_p^-(\tilde{\beta)}<\alpha/2$. The only way to find such a value is trial and error by moving progressively away from $\hat{\beta}$ and repeating flipscores tests to compute the corresponding p-values. We select an amplitude $\epsilon$ and set $\tilde{\beta}^{(0)} = \hat{\beta}$ and, for $i =1,2,\ldots$, $\tilde{\beta}^{(i+1)} = \tilde{\beta}^{(i)} - \epsilon$. We stop as soon as we find an $i$ such that $f_p^-\left(\tilde{\beta}^{(i+1)}\right) < \alpha/2$, and define $\tilde{\beta} = \tilde{\beta}^{(i+1)}$. We can then start the bisection search between $\tilde{\beta}$ and $\tilde{\beta}+ \epsilon$. 

In the first iteration, we move towards $\hat{\beta}$ by a quantity $\epsilon/2$. We repeat the test at the obtained point and compute the p-value $p$. If $p$ is larger than or equal to $\alpha/2$, our guess remains $\tilde{\beta}$; we half the step size and move away from $\hat{\beta}$. Otherwise, we have found a new conservative confidence bound and should move again toward $\hat{\beta}$ after halving the step size. We proceed iteratively until the step size becomes smaller than a tolerance value.
One may expect that once we find a value $\beta_0$ with $f_p^-(\beta_0)=\alpha/2$, our search is completed, but finding a single value value that satisfies the condition does not ensure that we have found all valid points in the sets in Equations (\ref{eq:test_inversion}) and (\ref{eq:bound_definition}). In practice, for the sake of conservativeness, rather than looking for the $\inf$ in \cref{eq:bound_definition}, we look for
\[
    \tilde{\beta}_L = \max \left\{\beta_0 < \hat{\beta}:f_p^-(\beta_0)<\alpha/2\right\}.
\]
The monotonicity of $f_p^-$ and the absence of a stopping condition other than reaching the tolerance step size make it that the difference between $\tilde{\beta}_L$ and $\beta_L$ is at most the tolerance. 

The selection of $\epsilon$ is crucial: a value too low for $\epsilon$ means that we need a large number of tests to find the initial $\tilde{\beta}$, with a lot of time and computational cost. If $\epsilon$ is too large, we have multiple problems. First, reaching an absolute tolerance requires many bisection steps, which is again computationally expensive. We solve this issue by making the tolerance relative to the initial amplitude $\epsilon$, but this means that starting with a large $\epsilon$ gives a shallow approximation of $\beta_L$. Second, in some models, such as logistic, moving to too extreme values for $\beta_0$ produces a degenerate model, which affects the computation of scores, so taking too large leaps away from $\hat{\beta}$ is not optimal. Our current suggestion for the choice of $\epsilon$ is 
\[
    \epsilon = \max\left\{z_{1-\alpha/2}  \hat{\sigma}_\beta,\ \hat{\beta} / 100,\ 0.2\right\}.
\]
The first element of the set is the semi-interval of the classical Wald-type confidence set, calculated as the product of the Gaussian quantile depending on the level of the test and the consistent estimate of the standard deviation of the parameter of interest \citep[see][Chapter 2]{salvan_modelli_2020}. It is used as a baseline to provide a first approximation of the width of the interval. At the same time, since the estimated standard deviation is too small in case of variance misspecification, we account for the scale of the estimate by including a term that depends on $\hat{\beta}$. We also provide an arbitrary baseline, in case the estimate $\hat{\beta}$ and the parametric semi-interval are both close to 0. Numerical experiments led to the choice of the values 100 and 0.2, which allow the method to reach $\tilde{\beta}^{(0)}$ in a few attempts. We then look for starting values up to a distance of $10\epsilon$ from $\hat{\beta}$, and otherwise assume that the null model is becoming degenerate and use infinite as confidence bound in the corresponding direction.

In the \textit{symmetric confidence interval}, we look for a positive value $\delta_C$ such that 
\[
    \delta_C = \sup\left\{\delta: f_p^-(\hat{\beta}-\delta) + f_p^+(\hat{\beta}+\delta) \geq \alpha\right\}.
\]
The bisection procedure and conservativeness measures are analogous to the equitailed case. At the end of the iteration, we define $\beta_L = \hat{\beta} - \delta_C$ and $\beta_U = \hat{\beta} + \delta_C$, so that we produce a confidence interval that is symmetric around the estimate. Algorithms~1~and~2, in Appendix B, show the pseudocode for the equitailed and symmetric bisection procedures for the search of the confidence bounds.

\section{Simulation}
\label{sec:simul}
We explore six simulation settings. We first fit three correctly specified models, linear, logistic, and Poisson models. Then, we simulate a false (overdispersed) Poisson model, where we sample data from a negative binomial distribution and fit a Poisson distribution. Finally, we consider two normal models, where we generate with heteroskedasticity, depending on either the covariate corresponding to the regression parameter we test, or on the nuisance covariate, and fit a linear model without accounting for the heteroskedasticity. For each case, we repeat 1000 experiments, varying $N$ between 25, 50, and 100. The confidence level is $1-\alpha = 0.95$.

We both consider the equitailed and the symmetric confidence intervals. We compare them with two Wald-type confidence intervals, one with the classical estimate of covariance and one where we use sandwich covariance \citep[e.g.][]{fay_small-sample_2001}, computed with the R package \texttt{sandwich} \citep{zeileis2006object,zeileis2020various}. In Figures (\ref{fig:simul_cov_prob}) and (\ref{fig:simul_med_width}) we show the performance of the methods in terms of coverage probability of the true value (the higher the better) and median width of the confidence intervals across the 1000 experiments (the lower the better), respectively.

    
\begin{figure}
    \centering
    \includegraphics[width=\linewidth]{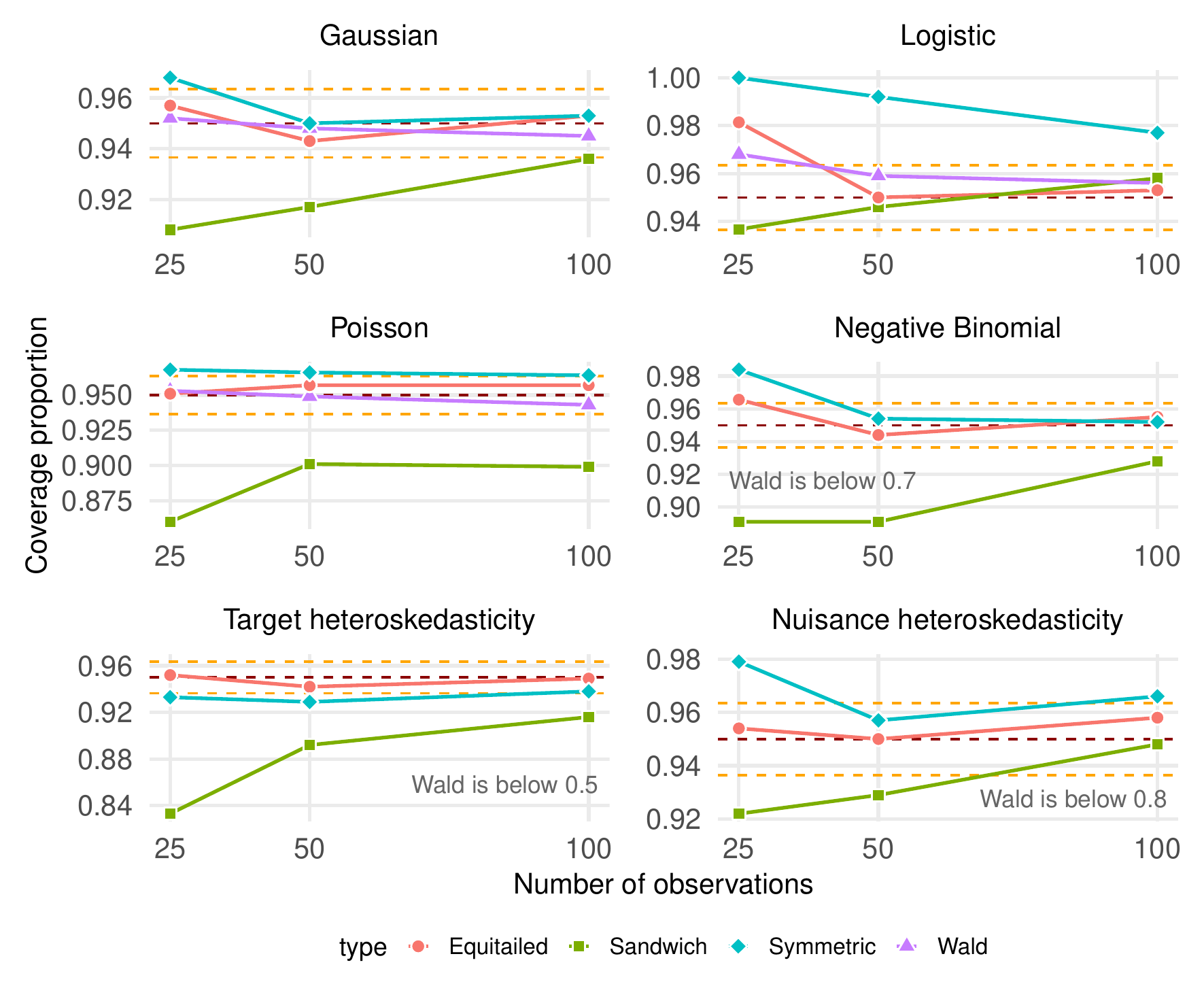}
    \caption{Coverage probabilities of the confidence intervals built from inversion of flipscores tests (equitailed and symmetric) and from parametric methods (Wald and sandwich), in 6 simulation settings, as the number of observations increases.}
    \label{fig:simul_cov_prob}
\end{figure}

\begin{figure}
    \centering
    \includegraphics[width=\linewidth]{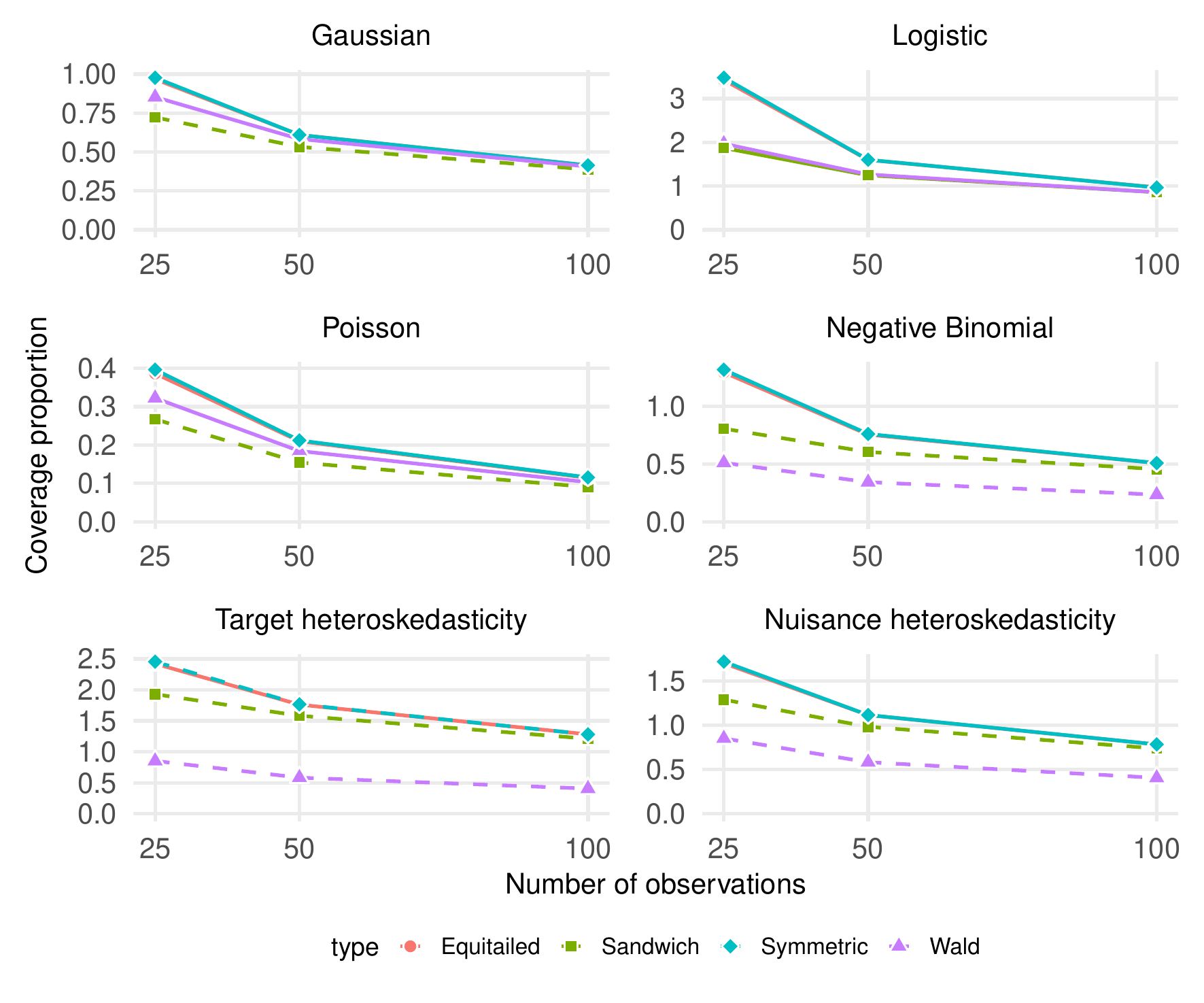}
    \caption{Median of the interval width of the confidence intervals built from inversion of flipscores tests (equitailed and symmetric) and from parametric methods (classical Wald and sandwich), in 6 simulation settings, as the number of observations increases. A dashed line indicates that the method does not reach nominal coverage in the setting.}
    \label{fig:simul_med_width}
\end{figure}

When the model is correctly specified, our methods and the classical Wald confidence intervals reach the nominal coverage proportion, while sandwich-based intervals show a low coverage probability, especially for low $N$ or in the Poisson case. Parametric models show a smaller interval width, as expected, but the difference reduces as $N$ increases. The symmetric flip-based interval shows a larger coverage proportion than the equitailed one, especially when $N$ is small, with a small cost in terms of the interval width. In case of the overdispersed Poisson model, the flip-based confidence intervals are the only ones with coverage probability at the nominal level. When we introduce undetected target heteroskedasticity, the equitailed flip-based confidence interval always shows nominal coverage, while the symmetric is below the nominal interval (0.9365, 0.9635) for $N = 25, 50$, with coverage probabilities of 0.933 and 0.929 respectively. It falls again in the nominal interval for $N = 100$. In this case, both parametric methods fail to reach nominal coverage. When heteroskedasticity is on the nuisance parameter, instead, the sandwich-based interval is able to account for it, at least for $N = 100$. Classical Wald-type intervals do not reach nominal coverage, whereas flip-based intervals do for all $N$. In these misspecified cases, the interval widths of the parametric methods are much lower than the flip-based ones, but it is not relevant when they do not reach the nominal coverage. In the only case where the sandwich-based interval reaches such coverage, that is, nuisance heteroskedasticity with $N=100$, the difference in its median width with those of the flip-based methods gets close to 0.

\section{Application}

High-throughput genomic data are well known to exhibit substantial deviations from idealized parametric assumptions. In particular, overdispersion, heteroskedasticity, and unobserved sources of biological and technical variability are widespread, making reliable uncertainty quantification a central challenge.

We consider differential expression analysis for Liver Hepatocellular Carcinoma using data from The Cancer Genome Atlas \citep[TCGA-LIHC;][]{tomczak_review_2015,erickson2016cancer}. After standard preprocessing, the dataset contains expression measurements for more than 18,000 genes in 344 patients. Each gene expression level is modeled as a function of the tumor stage, in the form of a binary indicator of whether the tumor is in the first pathological stage or in a more advanced stage. The inferential target is the regression coefficient associated with this variable, and the uncertainty is quantified through confidence intervals for this parameter. The gender and age of the patient act as nuisance covariates.

Because the true data-generating mechanism is unknown, we fit both the Poisson and Negative Binomial models. These two choices are standard in RNA-seq analysis, but they rely on markedly different assumptions about the mean–variance relationship. Rather than selecting a single ``best" model, we explicitly compare the resulting confidence intervals to assess their coherence under potential misspecification. We contrast our proposed nonparametric confidence intervals with classical Wald-type intervals constructed using (i) the model-based covariance estimator and (ii) the sandwich covariance estimator, the latter being asymptotically robust but still dependent on correct specification of the mean structure.

The distribution of the amplitudes is shown in \cref{fig:density_log_amplitudes}. The distributions of the sandwich- and flip-based intervals are very similar to each other, while the Wald-type intervals with classical covariance are different. The model that assumes a Poisson distribution is the one that suffers the most from the wrong specification and produces very short confidence intervals that do not ensure nominal coverage (as shown in the simulations). This behavior is consistent with the simulations in \cref{sec:simul} and with well-documented overdispersion in gene expression data: when the Poisson variance assumption is violated, model-based Wald intervals tend to underestimate uncertainty and can lead to overly optimistic conclusions \citep{oberg_technical_2012,nemcova_unjustified_2025}.

    
\begin{figure}
    \centering
    \includegraphics[width=0.6\linewidth]{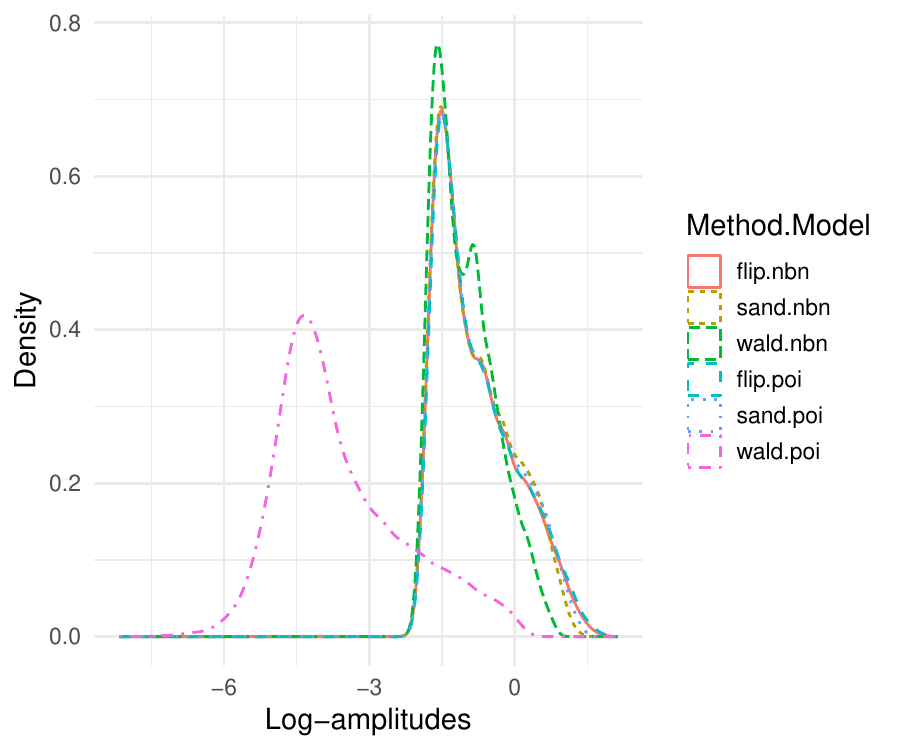}
    \caption{Distributions of the amplitudes of the confidence intervals according to construction method and model specification.}
    \label{fig:density_log_amplitudes}
\end{figure}

From now on, we focus on the comparison between the proposed flip-based intervals and the sandwich-based ones. \cref{fig:amplitude_scatter} compares the amplitudes of such confidence intervals for each gene. The flip-based intervals are generally wider, as also observed in the simulations, because they ensure nominal coverage. The effect appears to be more evident the larger the intervals are, that is, the more uncertainty associated with the estimated parameter.

    
\begin{figure}
    \centering
        \includegraphics[width=\linewidth]{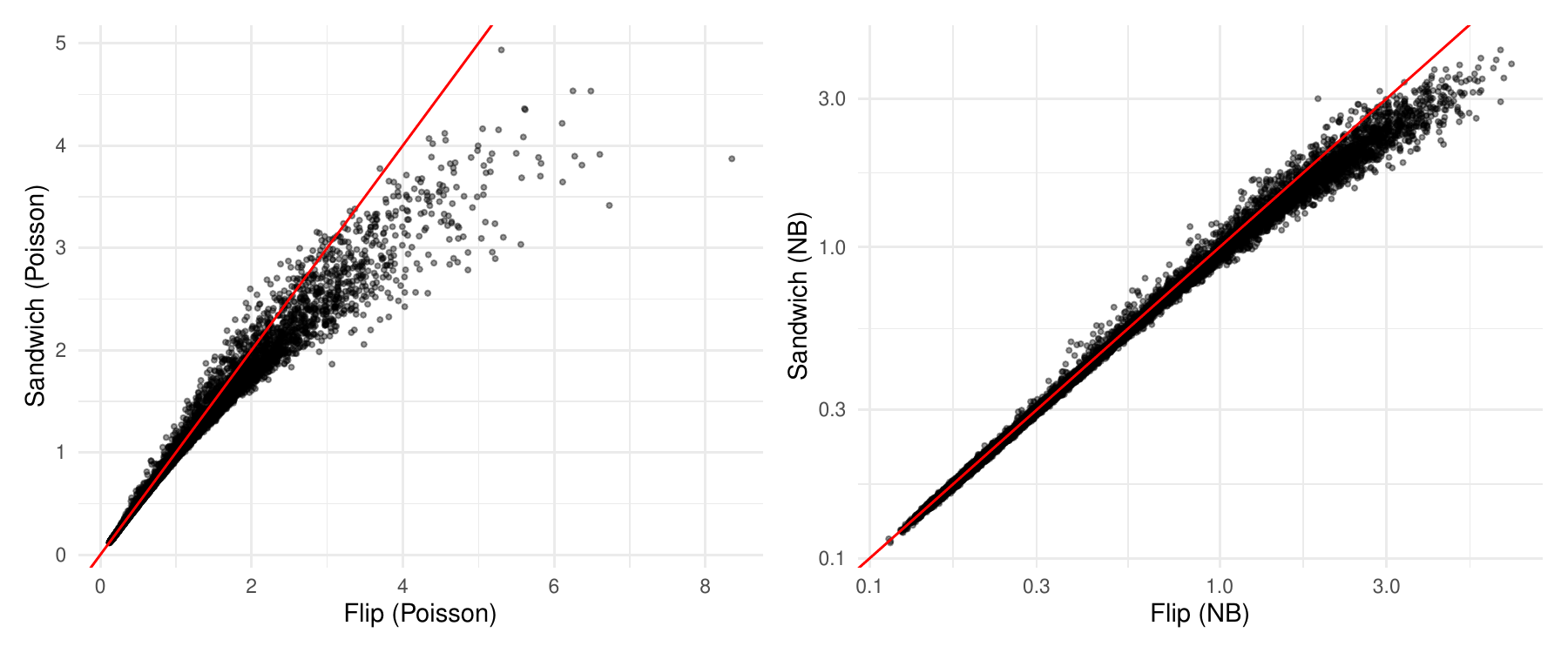}
    \caption{Amplitude comparison between flip and sandwich-based intervals for each gene.}
    \label{fig:amplitude_scatter}
\end{figure}

As a measure of stability of the constructed confidence intervals, for each method, we consider the overlap between the intervals built under Poisson and the negative binomial model specifications. The overlap is defined as
\begin{equation}
    \frac{2\max\{0,\ \min(\beta_U^{(NB)}, \beta_U^{(Poi)}) - \max(\beta_L^{(NB)}, \beta_L^{(Poi)})\}}{(\beta_U^{(NB)} - \beta_L^{(NB)})+(\beta_U^{(Poi)} - \beta_L^{(Poi)})}
    \label{eq:def_overlap}
\end{equation}
and ranges between 0 (in case of no overlap) and 1 (exact same intervals). The direct comparison of the overlaps between the flip and sandwich-based intervals is shown in \cref{fig:overlap_scatter}. We see that most of the points lie below the diagonal, underscoring how the flip-based method produces more stable intervals, in general.

    
\begin{figure}
    \centering
    \includegraphics[width=0.6\linewidth]{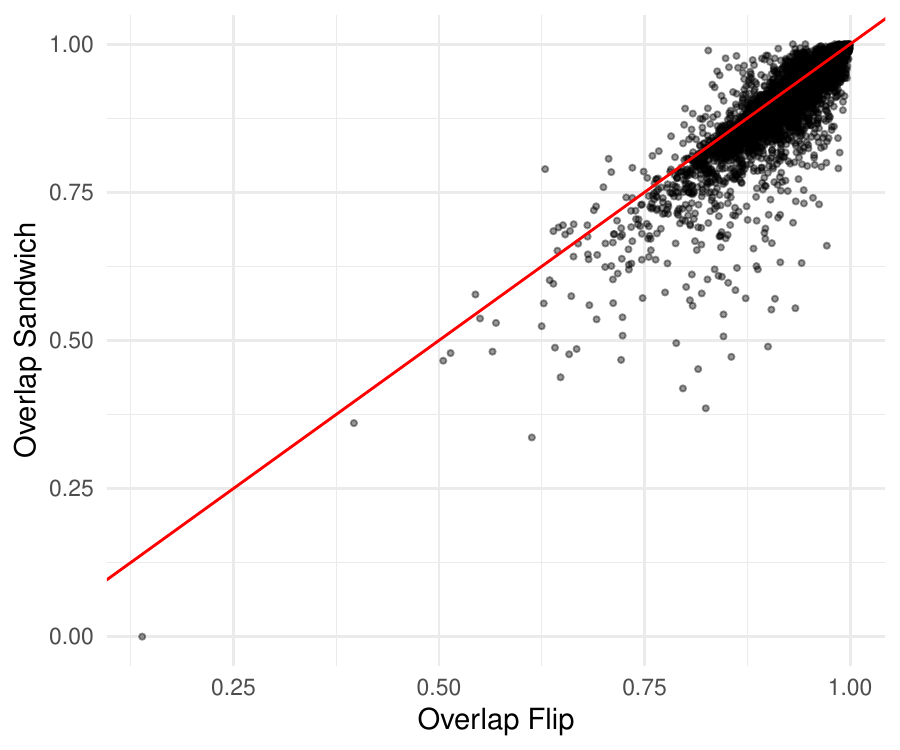}
    \caption{Comparison of the overlaps between the flip and sandwich-based methods for each gene.}
    \label{fig:overlap_scatter}
\end{figure}

A possible criticism might be that the overlap is larger for flip-based methods simply because the produced intervals are in general wider. \cref{fig:overlap_vs_length} shows the comparison of the overlaps as the lengths of the intervals vary, contesting this objection.

    
\begin{figure}
    \centering
    \includegraphics[width=0.6\linewidth]{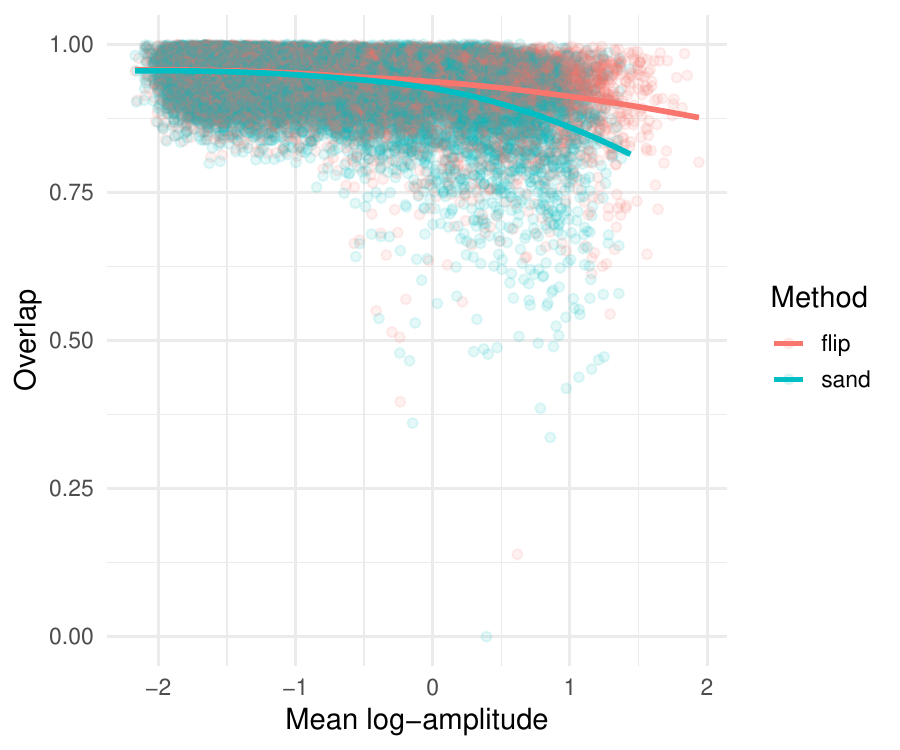}
    \caption{Overlap between the confidence intervals built with Poisson and negative binomial specifications, against the mean interval width of the two, using flip and sandwich-based methods. The overlap is larger for the flip-based method at all lengths.}
    \label{fig:overlap_vs_length}
\end{figure}

The observed coherence suggests that the proposed approach effectively mitigates the impact of variance misspecification, providing reliable uncertainty quantification without requiring correct specification of the full distribution.

\section{Discussion}
In this paper, we proposed a semi-parametric approach for constructing confidence intervals in generalized linear models. The method provides a proper way to obtain uncertainty quantification that is robust to variance misspecification, while retaining the familiar modeling framework of GLMs. We established theoretical guarantees for the validity of the procedure and complemented these results with simulation studies and a real data application. Both empirical investigations highlight that the proposed intervals maintain nominal coverage under various forms of heteroskedasticity and overdispersion, whereas classical parametric confidence intervals can exhibit substantial undercoverage when model assumptions are violated. By reducing sensitivity to arbitrary modeling choices, the method offers researchers the opportunity to provide reliable uncertainty inference, rather than point estimation alone, even in applications where modeling assumptions are difficult to verify, such as high trough-put genomics, limiting the possible damage induced by modeling mistakes.

The robustness of the proposed method comes directly from the properties of the sign-flip score tests \citep{hemerik2020robust,de2025inference}. As a consequence, the procedure relies on the correct specification of the mean structure. Extending the methodology to settings with certain forms of mean misspecification represents an interesting direction for future research. This would extend our range outside the field of glms, and more complex regression models tailored to specific applications could be incorporated, with the goal of preserving the robustness properties of the underlying tests.  In particular, mixed models \citep[see][]{andreella2025robust} and zero-inflated models, commonly used in the analysis of gene expression data, could be of great interest.

The present work focuses on inference for a single scalar parameter, avoiding so issues related to multiplicity. Extending the framework to simultaneous inference, for instance, through the construction of joint confidence regions or multiple confidence intervals, constitutes another natural direction for further investigation. Inevitably, in high-dimensional settings the loss of power may result in overly wide and uninformative intervals, but since the asymptotic properties of sign-flip score tests still hold for multivariate inference, the extension to moderately sized parameter sets is promisingly solid.

The proposed methodology is available through the R package \texttt{flipscores}, which implements the sign-flip testing procedures of \cite{hemerik2020robust} and \cite{de2025inference} and is available on Github. Confidence intervals can be obtained directly by applying the standard \texttt{confint} method to objects of class \texttt{flipscores}, facilitating the integration of the method into existing analysis workflows.


\section*{Acknowledgements}

This research was carried out with the support of a research grant funded by the University of Padova - Department of Statistical Sciences, under the BIRD 2023 funding scheme. The author Riccardo De Santis acknowledges the following funding, from the Italian Ministry of University and Research, PRIN 2022 - Project No.2022FZY9PM - CUP C53C24000740006. Generative AI tools (GPT-5) were used in the writing of this paper in the form of language improvement tools and coding assistants to produce clearer figures.\vspace*{-8pt}


\section*{Supporting Information}

Web Appendices are available with this paper.
Appendix A contains the proofs of Theorems~2~and~3 and Corollary~1. Appendix B provides the pseudocode to the bisection algorithms described in Section~3.2. Appendix C adds images and their interpretation to the application in Section~5.
\vspace*{-8pt}

\bibliography{references}
    
\clearpage
\appendix
\renewcommand{\theequation}{S.\arabic{equation}}
\renewcommand{\thefigure}{S.\arabic{figure}}
\renewcommand{\thetable}{S.\arabic{table}}
\setcounter{page}{1}

\section*{Supplementary to ``Robust confidence intervals for generalized linear models"}
\section{Proofs to theorems}
\subsection{Proof to Theorem 2}
\setcounter{theorem}{1}
\begin{theorem}
    Assume that the observations are independent and generated by the linear model 
    \[
        Y_i = x_i\beta + \mathbf{z}_i^T \boldsymbol{\gamma} + \varepsilon_i\,,\quad \varepsilon_i \sim \mathcal{N}(0, \sigma^2)\, ,
    \]
    for $i = 1,\ldots,n$.
    Let $\hat{\beta}_{\text{obs}}$ be the maximum likelihood estimate of $\beta$ and let $\beta_A < \beta_B < \hat{\beta}_{\text{obs}}$. Let $p_k$, $k=A,B$, be the p-values associated with the effective flipscores tests $H_0^{(k)}:\beta=\beta_k$ versus $H_1^{(k)}: \beta > \beta_k$, for $k = A,B$, respectively, assuming that the tests share the same set of flip matrices. Then, $p_A\leq p_B$.
\end{theorem}

\begin{proof}
    We fix $w$ flip matrices and test the null hypothesis $H_0^{(A)}:\ \beta = \beta_A$ versus the one-sided alternative $H_1^{(A)}: \beta > \beta_A$. We call $w_A$ the number of flip matrices that lead to a score that is lower than the observed one, that is, the test p-value will be $p_A = 1 - w_A / w$. Analogously, when testing the null hypothesis $H_0^{(B)}:\ \beta = \beta_B$, we will have $p_B$ and $w_B$. In order to have monotonicity, we want $p_A$ not to be larger than $p_B$, which is equivalent to proving $w_A \geq w_B$.
    
     Let $S^A(F_g)$ and $S^B(F_g)$ be the effective score related to the fixed flip matrix $F_g$ and, respectively, to the test with null $\beta_A$ and $\beta_B$. The observed scores are $S^A(I)$ and $S^B(I)$, where $I$ is the identity matrix of size $n$. It is sufficient to prove that if a flip matrix $F_g$ is such that $S^B(F_g) < S^B(I)$, then $S^A(F_g) < S^A(I)$. 
     In linear models, we can write $\hat{\boldsymbol{\mu}}$ of \cref{eq:eff_score_stat}, under the null hypothesis $\beta = \beta_k,\ k = A,B$, as $X\beta_k + H(\mathbf{y} - X\beta_k)$, where $H = Z\left(Z^TZ\right)^{-1}Z^T$ does not depend on the null hypothesis.
     In linear models, under the null hypothesis $\beta = \beta_k,\ k = 0,1$, it holds that $\hat{\boldsymbol{\mu}}_k=X\beta_k + H(\mathbf{y} - X\beta_k)$, where $H = Z\left(Z^TZ\right)^{-1}Z^T$ does not depend on the null hypothesis. Substituting $\hat{\boldsymbol{\mu}}_k$ in \cref{eq:eff_score_stat} and neglecting the constant term $n^{-1/2}$, which is common to all scores, we have that 
    \begin{align*}
        S^k(F_g) =&\, X^\top W_k^{1/2} (I - H) F_g V_k^{-1/2} (I - H)  Y \\
        &- X^\top W_k^{1/2} (I - H) F_g V_k^{-1/2}(I-H) X \beta_k\,.
    \end{align*}
    The matrices $W_k=V_k^{-1}$ are of the form $\hat{\sigma}^2_k I$, where $\hat{\sigma}^2_k$ is the variance estimated under the null and does not depend on the flip. We have then
\begin{align*}
    S^B(F_g) =& \frac{1}{\hat{\sigma}^2_B}\Big[ X^\top (I - H) F_g (I - H) Y \\
    &- X^\top (I - H) F_g (I-H) X \beta_B\Big]\,,\\
    S^B(I) =&\frac{1}{\hat{\sigma}^2_B}\left[ X^\top (I - H) Y - X^\top (I - H) X \beta_B\right]\,,
\end{align*}
and
\begin{align*}
    S^A(I) &= \frac{1}{\hat{\sigma}^2_1}\left[X^\top (I - H) Y - X^\top (I - H) X \beta_A\right] \\
    & = \frac{\hat{\sigma}^2_B}{\hat{\sigma}^2_A} S^B(I) - \frac{1}{\hat{\sigma}^2_A} X^\top (I - H) X \left(\beta_A - \beta_B\right)\,.
\end{align*}
Then, 
\begin{align*}
    S^A(F_g) =&\,\frac{1}{\hat{\sigma}^2_A}\Big[ X^\top (I - H)F_g(I - H) Y\\
    & \qquad \ - X^\top (I - H) F_g (I-H) X \beta_A\Big] \\
     =&\, \frac{\hat{\sigma}^2_B}{\hat{\sigma}^2_A} S^B(F_g) - \frac{1}{\hat{\sigma}^2_A} X^\top (I - H) F_g(I-H) X \big(\beta_A - \beta_B\big) \\
    < &\,\frac{\hat{\sigma}^2_B}{\hat{\sigma}^2_A} S^B(I) - \frac{1}{\hat{\sigma}^2_A} X^\top (I - H) F_g (I-H) X \big(\beta_A - \beta_B\big) \\
    = & \, S^A(I) + \frac{1}{\hat{\sigma}^2_A} X^\top (I - H) X \big(\beta_A - \beta_B\big) \\
    & \qquad \quad- \frac{1}{\hat{\sigma}^2_A} X^\top (I - H) F_g  (I-H) X \big(\beta_A - \beta_B\big) \\
    = &\, S^A(I) + \frac{1}{\hat{\sigma}^2_A} X^\top (I - H) (I - F_g) (I-H) X \big(\beta_A - \beta_B\big),
\end{align*}
using the hypothesis of $S^B(F_g) < S^B(I)$ and writing the idempotent matrix $(I-H)$ as $(I - H) I (I-H)$ for the last equivalence. The matrix $I-F_g$ is diagonal with only values of $0$ or $2$, and the product $(I - H) (I - F_g) (I - H)$ is positive semidefinite, being symmetric and with all non-negative eigenvalues. The quadratic form $X^\top (I - H) (I - F_g) (I - H) X$ is then non-negative and, since $\beta_A - \beta_B < 0$, the last addendum is non-positive, so $S^A(F_g) < S^A(I)$.
\end{proof}

\subsection{Proof to Corollary 1}
\setcounter{corollary}{0}
\begin{corollary}
    Assume that the observations are independent and generated by the linear model 
    \[
        y_i = x_i \beta + \mathbf{z}_i^T \boldsymbol{\gamma} + \varepsilon_i\,,\quad \varepsilon_i \sim \mathcal{N}(0, \sigma^2)\, ,
    \]
    for $i=1,\ldots,n$. 
    Let $\hat{\beta}_{\text{obs}}$ be the estimate of $\beta$ and let $\beta_A < \beta_B < \hat{\beta}_{\text{obs}}$. Let $p_k$, $k=A,B$, be the p-values associated with the standardized flipscore tests $H_0^{(k)}:\beta=\beta_k$ versus $H_1^{(k)}: \beta > \beta_k$, for $k = A,B$, respectively. Then, as $n\rightarrow\infty$, it holds that $p_A\leq p_B$.
\end{corollary}
\begin{proof}
    We can follow the same proof of \cref{thm:monot_lm_eff}, substituting any $S^k(F_g)$ with the corresponding standardized version ${S^{k}}^*(F_g)$ defined in \cref{eq:def_standardization}. We should notice that in linear models $\mathbb{V}(S^B(F_g)) =\mathbb{V}(S^A(F_g))=:\mathbb{V}_{F_g}$. Assuming ${S^{B}}^*(F_g) < {S^{B}}^*(I)$, we that ${S^{A}}^*(F_g)$ is less than
    \[
        {S^A}^*(I)+ \frac{1}{\hat{\sigma}^2_A} X^\top (I - H) \left(\frac{I}{\mathbb{V}_I^{1/2}} - \frac{F_g}{\mathbb{V}^{1/2}_{F_g}}\right) (I-H) X \left(\beta_A - \beta_B\right).
    \]
    From \cref{eq:asy_equi_var_flipscores} and with the same reasoning as at the end of the proof of \cref{thm:monot_lm_eff}, we obtain ${S^{A}}^*(F_g) \lesssim {S^{A}}^*(I)$.
\end{proof}
\subsection{Proof to Theorem 3}
\begin{theorem}
Assume that the observations are independent and generated from a generalized linear model defined by Equations (\ref{eq:model})--(\ref{eq:glm_link}). Let $\hat{\beta}$ denote the estimate of $\beta$ and let $\beta_A < \beta_B < \hat{\beta}$. Let $S^{(k)}(F_g)$ denote the effective or standardized sign-flip score statistic computed under $H_0^{(k)}:\beta=\beta_k$ with flip matrix $F_g$, and let $S^{(k)}(I)$ denote the observed statistic (identity flip). 
Then, for any fixed flip matrix $F_g$,
\begin{equation*}
    \Pr\left(S^{(A)}(F_g) < S^{(A)}(I)\mid S^{(B)}(F_g) < S^{(B)}(I)\right) \;\longrightarrow\;1
\end{equation*}
as $n\to\infty$. Consequently, the sign-flip p-value for testing $H_0^{(A)}$ is asymptotically no larger than that for testing $H_0^{(B)}$.
\end{theorem}

\begin{proof}
Let $S(\beta)$ denote the score and $I_n(\beta)$ the Fisher information. By the score equation $S(\hat{\beta}_n)=0$ and the local asymptotic linearity assumption, we obtain the expansion
\begin{equation*}
    S(\beta) = -(\hat{\beta}_n-\beta) I_n(\beta_0) + o_p(\sqrt n).
\end{equation*}
Since $I_n(\beta_0)=O(n)$, the standardized score statistic (\ref{eq:def_standardization}) satisfies, for $k = A,B$,
\begin{equation*}
    {S^k}^*(I) = c_k \sqrt n\,(\beta_k-\hat{\beta}_n) + W_n + o_p(1),
\end{equation*}
where $W_n \Rightarrow N(0,1)$ and $c_k>0$ is a constant. In particular, $c_k = {I_n(\beta_0)}/{{I_n(\beta_k)}^{1/2}}$. As $n$ grows, $\hat{\beta}_n\rightarrow\beta_0$, so $I_n(\beta)/n \rightarrow I(\beta)$ and we obtain $c_k = {n}^{1/2}{I(\beta_0)}/{{I(\beta_k)}^{1/2}} + o({n}^{1/2})$. Consider now $\beta_k$ in a local neighborhood of $\hat{\beta}$, say $\beta_k = \hat{\beta} + O(n^{-1/2})$. Then, $I(\beta_k) = I(\beta_0) + o(1)$, so $c_k = {n}^{1/2}{I(\beta_0)}^{1/2} + o({n}^{1/2})$ which means that all $c_k$ are asymptotically equal. Hence, the statistic has an asymptotically linear drift away from $\hat{\beta}$. In particular, if $\beta_A < \beta_B < \hat{\beta}$ then
\begin{equation*}
    |{S^A}^*(I)| > |{S^B}^*(I)| + o_p(1).
\end{equation*}
Consider now a statistic ${S^k}^{*}(F_g)$ for a generic $F_g$. By the self-normalized and multiplier central limit theorems,
\begin{equation*}
    {S^{k}}^*(F_g) \xrightarrow{d} N(0,1).
\end{equation*}
Thus, the flipped statistics are asymptotically centered and do not contain the deterministic drift present in the observed score. Because the observed statistics satisfy ${S^{A}}^*(I) < {S^{B}}^*(I)$ with probability tending to 1, while the flipped statistics remain $O_p(1)$, it follows that
\begin{equation*} 
    \Pr \left({S^{A}}^*(F_g) < {S^{A}}^*(I)\mid S^{(B)}(F_g) < S^{(B)}(I) \right) \longrightarrow 1.
\end{equation*}
The statement about the ordering of p-values follows immediately from their definition in the end of Section~2, 
and the statement for the effective score follows from the asymptotic coincidence of the p-values between effective and standardized tests (consequence of \cref{eq:asy_equi_var_flipscores}).
\end{proof}
\section{Bisection algorithms}

\cref{alg:equitailed,alg:symmetric} provide the pseudocode to the procedures described in Section 3.2.
\begin{algorithm}[htp]
\caption{Equitailed Confidence Interval via Flipscores Test Inversion}
\label{alg:equitailed}
\KwIn{
    $mod$: a fitted flip-scores model;\quad $dat$: data object;\quad $\beta_0$: previous point of the bisection iteration; \quad $\beta_C$: a current conservative bound; \quad $low$: $1$ if we look for the lower bound, $-1$ for the upper; \quad $move\_dir$: movement direction, $1$ if moving right, $-1$ if left; \quad $\epsilon$: previous bisection step-size; \quad $tol$: convergence tolerance; \quad $\alpha$: significance level.
}
\SetKwFunction{FConfBoundEquit}{ConfBoundEquit}
\SetKwProg{Fn}{Function}{:}{}
\Fn{\FConfBoundEquit{$mod, dat, \beta_0, \beta_C, low, move\_dir, \epsilon, tol, \alpha$}}{
    $\epsilon \gets \epsilon / 2$\;
    $\beta_0 \gets \beta_0 + move\_dir \times \epsilon$\;

    \eIf{$lower == 1$}{
        $alternative \gets$ ``greater''\;
    }{
        $alternative \gets$ ``less''\;
    }

    Compute $p \gets$ \textsc{PermutationTest}$(mod, dat, \beta_0, alternative)\$p.value$\;

    \uIf{$p < \alpha/2$}{
        $\beta_C \gets \beta_0$; \tcp*[f]{Conservative bound found}
    }
    \uIf{$\epsilon < tol$}{
        \Return{$\beta_C$; \tcp*[f]{Return the current bound}}
    }
    \uElseIf{$p < \alpha/2$}{
        \tcp{Move toward $\hat{\beta}$ to increase $p$}
        \Return{\FConfBoundEquit{$mod, dat, \beta_0, \beta_C, low, low, \epsilon, tol, \alpha$}}\;
    }
    \Else{
        \tcp{Move away from $\hat{\beta}$ to decrease $p$ or find other $\beta_0$ with $p=\alpha/2$}
        \Return{\FConfBoundEquit{$mod, dat, \beta_0, \beta_C, low, -low, \epsilon, tol, \alpha$}}\;
    }
}
\end{algorithm}

\begin{algorithm}[htp]
\caption{Symmetric Confidence Interval via Permutation Test Inversion}
\label{alg:symmetric}
\KwIn{
    $mod$: a fitted flip-scores model;\quad $dat$: data object;\quad $\hat{\beta}$: estimate of the parameter of interest; \quad $\delta$: previous half-width of the bisection iteration; \quad $\delta_C$: a current conservative half-width; \quad $move\_dir$: movement direction, $1$ if moving away from $\hat{\beta}$, $-1$ if moving toward; \quad $\epsilon$: previous bisection step-size; \quad $tol$: convergence tolerance; \quad $\alpha$: significance level.
}
\SetKwFunction{FConfBoundSymm}{ConfBoundSymm}
\SetKwProg{Fn}{Function}{:}{}
\Fn{\FConfBoundSymm{$mod, dat, \hat{\beta}, \delta, \delta_C, move\_dir, \epsilon, tol, \alpha$}}{
    $\epsilon \gets \epsilon / 2$\;
    $\delta \gets \delta + move\_dir \times \epsilon$\;

    $\beta_L \gets \hat{\beta} - \delta$; \quad $\beta_U \gets \hat{\beta} + \delta$\;

    Compute $p_L \gets$ \textsc{PermutationTest}$(mod, dat, \beta_L, ``greater'')\$p.value$\;
    Compute $p_U \gets$ \textsc{PermutationTest}$(mod, dat, \beta_U, ``less'')\$p.value$\;
    $p_{sum} \gets p_L + p_U$\;

    \uIf{$p_{sum} < \alpha$}{
        $\delta_C \gets \delta$; \tcp*[f]{Conservative interval found}
        }
    \uIf{$\epsilon < tol$}{
        \Return{$\delta_C$;\tcp*[f]{Return conservative interval}}
    }
    \uElseIf{$p_{sum} < \alpha$}{
        \tcp{Total $p$ too small $\Rightarrow$ reduce interval width}
        \Return{\FConfBoundSymm{$mod, dat, \hat{\beta}, \delta, \delta, -1, \epsilon, tol, \alpha$}}\;
    }
    \Else{
        \tcp{Total $p$ too large or equal to $\alpha$ $\Rightarrow$ increase interval width}
        \Return{\FConfBoundSymm{$mod, dat, \hat{\beta}, \delta, \delta_C, 1, \epsilon, tol, \alpha$}}\;
    }
}
\end{algorithm}

\section{Further images to the application}

Figure~5, in the main text, showed how flip-based intervals are generally wider than sandwich-based counterparts. This is underscored by \cref{fig:amplitude_histogram}, where the distributions of the log-ratios between the amplitudes of the flip-based and sandwich-based intervals are right-skewed and with positive means of 0.024 (Poisson) and 0.03 (negative binomial). 

\begin{figure}
    \centering
    \includegraphics[width = \linewidth]{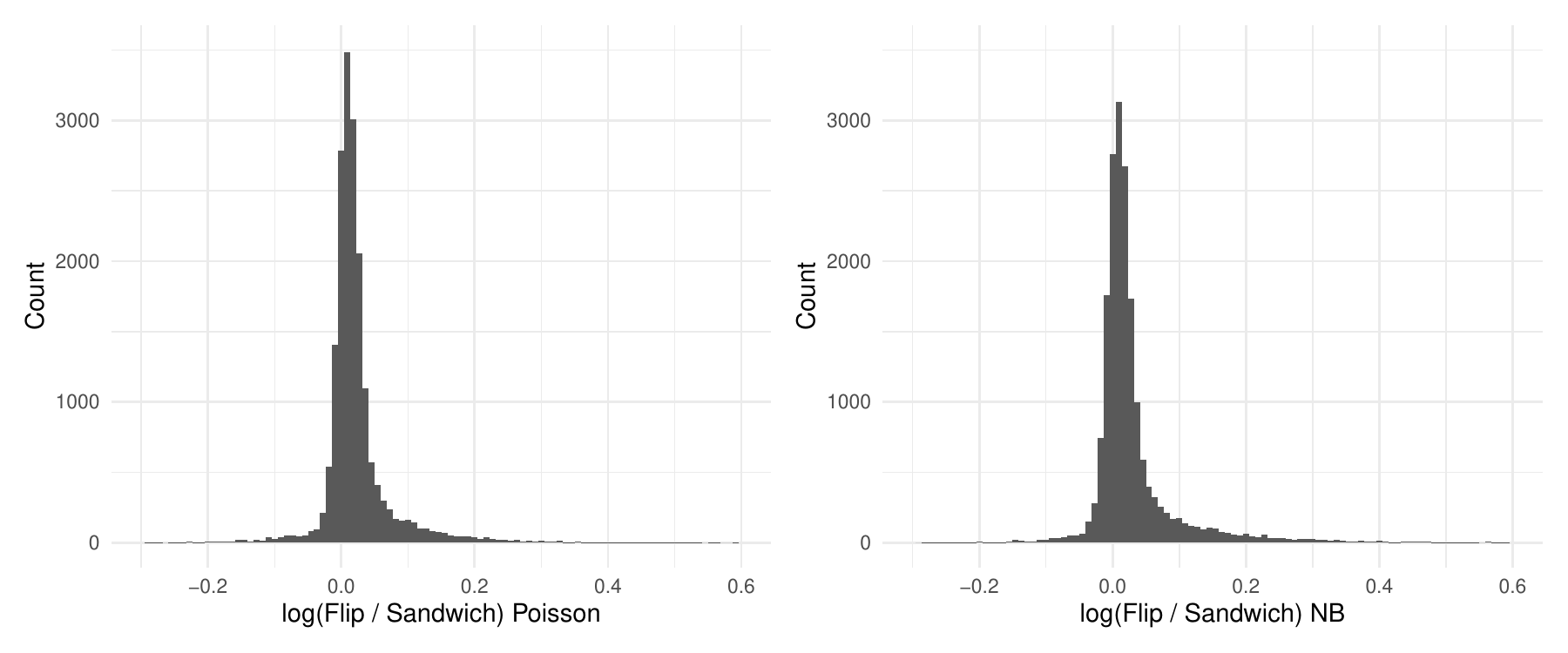}
    \caption{Distributions of the log-ratios of the amplitudes of the flip-based and sandwich-based intervals.}
    \label{fig:amplitude_histogram}
\end{figure}

\cref{fig:overlap_histograms} shows the distribution of the overlaps across the 18000+ genes. The flip-based distribution is more skewed towards the maximum value of 1 than the sandwich distribution. The Wald distribution is almost completely to the left due to the small width of the Poisson-based intervals.
\begin{figure}
    \centering
    \includegraphics[width=0.6\linewidth]{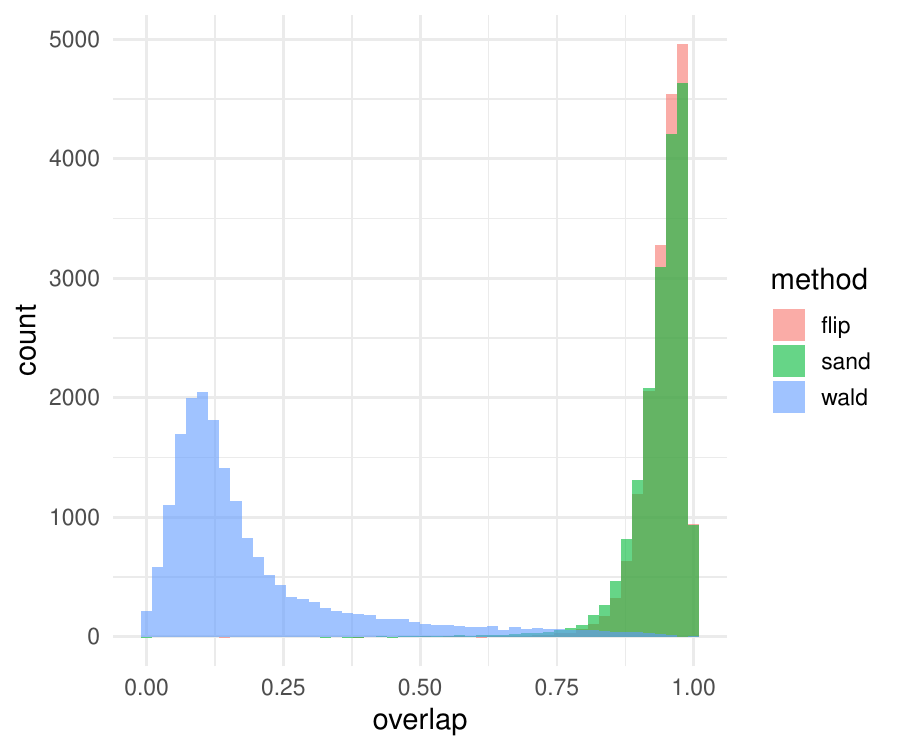}
    \caption{Distribution of the overlap of the confidence intervals built with each methods. The overlap defined in \cref{eq:def_overlap} is computed between the amplitudes of the intervals assuming Poisson and negative binomial model specifications.}
    \label{fig:overlap_histograms}
\end{figure}

\end{document}